\documentclass[10pt,letterpaper]{article}
\usepackage[utf8]{inputenc}

\usepackage{textcomp,marvosym}

\usepackage{fixltx2e}

\usepackage{amsmath,amssymb}

\usepackage{cite}

\usepackage{nameref,hyperref}

\usepackage{microtype}
\DisableLigatures[f]{encoding = *, family = * }

\usepackage{rotating}

\usepackage{graphicx}
\usepackage[usenames,dvipsnames]{color}
\usepackage{subfigure}
\usepackage{wrapfigure}
\usepackage{algorithmic}
\usepackage{algorithm}
\usepackage{listings}
\usepackage{epstopdf}
\usepackage{amsthm}
\usepackage{times}
\usepackage{xspace}
\usepackage{soul} 
\usepackage{enumitem}
\usepackage{float}
\usepackage{url}
\usepackage{ulem} 
\usepackage{listings}

\usepackage[final]{changes}
\usepackage{lipsum}% <- For dummy text

\usepackage[aboveskip=1pt,labelfont=bf,labelsep=period,justification=raggedright,singlelinecheck=off]{caption}

\makeatletter
\renewcommand{\@biblabel}[1]{\quad#1.}
\makeatother

\date{}

\newcommand{\junk}[1]{}

\newcommand{\fedor}{FEDoR\xspace}
\makeatletter
\newtheorem*{rep@theorem}{\rep@title}
\newcommand{\newreptheorem}[2]{%
\newenvironment{rep#1}[1]{%
 \def\rep@title{#2 \ref{##1}}%
 \begin{rep@theorem}}%
 {\end{rep@theorem}}}
\makeatother

\newcommand{\E}{\mathrm{E}}

\newtheorem{theorem}{Theorem}
\newreptheorem{theorem}{Theorem}
\newtheorem{lemma}{Lemma}
\newreptheorem{lemma}{Lemma}
\newtheorem{defn}{Definition}

\begin{document}
%\vspace*{0.35in}

% Please capitalize all terms in the title except conjunctions, prepositions, and articles.
\begin{flushleft}
{\Large
\textbf\newline{A Mechanism for Fair Distribution of  Resources without Payments}
}
\newline
\\
Evgenia Christoforou\textsuperscript{1,2,*},
Antonio Fern\'andez Anta\textsuperscript{1},
Agust\'in Santos\textsuperscript{1,3}
\\
\bigskip
\bf{1} IMDEA Networks Institute, Madrid, Spain
\\
\bf{2} Universidad Carlos III, Madrid, Spain
\\
\bf{3} Universidad Rey Juan Carlos, Madrid, Spain
\\
\bigskip

% Use the asterisk to denote corresponding authorship and provide email address in note below.
* evgenia.christoforou@imdea.org

\end{flushleft}
\section*{Abstract}
\deleted[]{We have designed a mechanism for Fair and Efficient Distribution of Resources (\fedor) in the presence of strategic agents. We consider a multiple-instances, Bayesian setting, where in each round the preference of an agent over the set of resources is a private information. We assume that $n$ agents are competing for $k$ non-identical indivisible goods, ($n>k$), in each of the $r$ rounds. In each round the strategic agents declare how much they value receiving any of the goods in the specific round. The agent declaring the highest valuation receives the good with the highest value, the agent with the second highest valuation receives the second highest valued good, etc. 
Hence we assume a decision function that assigns goods to agents based on their valuations. 
The novelty of the mechanism is that no payment scheme is required to achieve truthfulness in a setting with rational/strategic agents. The \fedor mechanism takes advantage of the repeated nature of the framework, and through a statistical test is able to punish the misreporting entities and be fair, truthful and socially efficient. \fedor is fair in the sense that, in expectation over the course of the rounds, all agents will receive the same good the same amount of times. \fedor is an eligible candidate for applications that require fair distribution of resources over time. For example, equal share of bandwidth for nodes through the same point of access. But further on, \fedor could be applied in less trivial settings like sponsored search, where payment is necessary and could be given in the form of a flat participation fee. \fedor could be a good candidate in a setting like that to solve the problem of starvation of publicity slots for some advertisers that have a difficult time determining their true valuations. To this extent we have performed a comparison with traditional mechanism applied to sponsored search, presenting the advantage of \fedor.   }
\added[]{We design a mechanism for Fair and Efficient Distribution of Resources (\fedor) in the presence of strategic agents. We consider a multiple-instances, Bayesian setting, where in each round the preference of an agent over the set of resources is a private information. We assume that in each of $r$ rounds $n$ agents are competing for $k$ non-identical indivisible goods, ($n>k$). In each round the strategic agents declare how much they value receiving any of the goods in the specific round. The agent declaring the highest valuation receives the good with the highest value, the agent with the second highest valuation receives the second highest valued good, etc. 
Hence we assume a decision function that assigns goods to agents based on their valuations. 
The novelty of the mechanism is that no payment scheme is required to achieve truthfulness in a setting with rational/strategic agents. The \fedor mechanism takes advantage of the repeated nature of the framework, and through a statistical test is able to punish the misreporting agents and be fair, truthful, and socially efficient. \fedor is fair in the sense that, in expectation over the course of the rounds, all agents will receive the same good the same amount of times. \fedor is an eligible candidate for applications that require fair distribution of resources over time. For example, equal share of bandwidth for nodes through the same point of access. But further on, \fedor can be applied in less trivial settings like sponsored search, where payment is necessary and can be given in the form of a flat participation fee. \fedor can be a good candidate in a setting like that to solve the problem of starvation of publicity slots for some advertisers that have a difficult time determining their true valuations. To this extent we perform a comparison with traditional mechanisms applied to sponsored search, presenting the advantage of \fedor.}

%\linenumbers

\section*{Introduction}

\added[]{Resource allocation in an efficient and fair manner among strategic agents~\footnote{We use the terms agent and player interchangeably},that can misreport their values to increase their benefit, is a non-trivial problem. }
\deleted[]{Consider a scenario where strategic agents need to have resources allocated to them in the most efficient and fair manner.} The most straightforward approach is the design of mechanisms where, once the  solution concept is achieved, fairness and efficiency are guaranteed. In order for the designed mechanism to be incentive compatible, \deleted[]{a transfer function} \added[]{economic incentives/payments} must usually be designed. In this work, we do not use \deleted[]{a transfer function}\added[]{any economic incentives/payments} and hence our mechanism design goals are achieved without any incentive constraint, that is without money. There is a plethora a real world applications deriving from computer science or from social sciences were payments can not be used as incentives for an allocation mechanism. 

In the field of computer science, \fedor can be applied to solve the problem of efficiently allocating CPU cycles between processes. Another application could be the periodical allocation in a fair manner of bandwidth among nodes connected to the same point of access. Imagine the scenario where many nodes are connected to the same point of access. An application could be built to support the feature of nodes declaring how much they value receiving bandwidth from the point of access at a specific time interval. Then, nodes would be served from the point of access according to \fedor . This would guarantee fairness and efficiency over time for all the nodes. More \deleted[]{application can be thought} \added[]{applications can be imagined}, where \fedor could allocate resources over the cloud. In this case it depends on the policy and the specific implementation on how \fedor could apply. \added[]{For example, a number of users with only periodic demands for large amount of resources could be sharing resources using \fedor instead of resolving into costly provisioning of resources.}     

Another possible application is hospitals sharing valuable equipment. No allocation based on the money that each hospital is willing to give for renting the equipment can be made. Hence, \fedor is a perfect candidate, without monetary payments can guarantee that all the hospitals sharing the valuable equipment will be allocated the same machine the same amount of time in expectation and most importantly in an efficient manner (when they most need it). Similar application could also be good candidates for \fedor , as long of course no life threatening decisions are left to be taken by the mechanism. 

Based on the nature of the problem, \deleted[]{we are assuming} \added[]{we assume} a Bayesian setting where agents have private information. Agents can be cheating, declaring preferences that do not correspond to their true valuations in an effort to increase their utility.  Unlike mechanisms with payments, were money incentives enforce the good behavior of the agents, this is not the case here. \fedor is designed especially to detect cheating and punish it. \added[]{This is feasible since \fedor is designed to work in a repeated setting, which allows the mechanism to gain knowledge on the agents' valuations and appropriately punish misbehaving agents, incentivizing their good behavior. Using the concept of linking mechanisms (explained in the Previous and Related Work paragraph) the mechanism is able to guarantee fairness and social efficiency.} 

We are assuming that preferences are declared within a common unit measure. In realistic applications, valuations might not even appear within the same unit of measure. Agents might not be able to associate their valuation under a common unit of measure.\deleted[]{ \fedor handles this problem by applying the Probability Integral Transformation (PIT)  and transforming the reported valuations to the uniform [0,1] distribution. Hence, whichever distribution of valuations can transform  to the uniform [0,1] distribution. This is why from this point onwards we will assume implicitly that the valuations declared by the agents follow the uniform [0,1] distribution.}
\added[]{To deal with this matter \fedor uses a common normalization process called Probability Integral Transformation (PIT)~\cite{pitref}. Using PIT any cumulative probability distribution function can be transformed to the uniform [0,1] distribution. 
Essentially, the PIT of a value $x$ extracted from a probability distribution is the aggregated probability of the values no larger than $x$ in the distribution. Intuitively, it is similar to the percentile of $x$. 
Under the assumption that \fedor has a historical sample of the agents' valuations, a distribution fitting the agents valuations can be defined. Since knowing the agent's valuation distribution is the only thing that we need to know to apply the PIT, from this point onwards we will assume that the valuations declared by the agents follow the uniform [0,1] distribution. In the work of Santos et al.~\cite{QPQ,QPQarxiv}, this approach was used in an analogous way to transform the agent's cost to the uniform distribution (see Figure 1. in~\cite{QPQ}). Having transformed the agent's valuation to the uniform [0,1] distribution, a {goodness of fit} (GoF) test can define whether an agent's valuation is following its true distribution, and act in an analogous manner; this guarantees that \fedor is a truthful mechanism (as we prove below).}

%%%%%%%%%%%%%%%%%%%%%%%%%%%%%%%%%%%%%%%%%%%%%%%%%%%%%%%%%%%%%%%%%%%%%%
%%%%%%%%%%%%%%%%%%NON-TRIVIAL MOTIVATING EXAMPLE%%%%%%%%%%%%%%%%%%%%%%

\added[]{
\paragraph{Non-trivial motivating example} We have mentioned before a number of intuitive examples where \fedor could apply. What these examples have in common is the difficulty or inappropriateness of using payments to allocate the resources. This fact, in combination with the repeated nature of the setting makes \fedor a perfect candidate in those scenarios. The good properties of the mechanism though (fairness, social efficiency and truthfulness) can make it a suitable candidate also to scenarios in which payments are essential for resource allocation, for example to sponsored search. }

\added[]{
As the business model of search engines' (i.e., Google, Yahoo!, Bing, etc) has evolved over time their main source of income are the ad links delivered to users each time they do a search. For each keyword search, there exists a number of available slots for advertisement in each page of results. Due to the fact that slots are limited while advertisers are numerous, and in combination with the search engines' goal for revenue maximization, auction mechanisms are used for slot allocation. 
}

\added[]{
The setting of sponsored search fits the setting for which \fedor was designed. A number of slots (goods in the general version) are allocated over $r$ consecutive instances, having the same advertisers (in a simplified version of the problem) competing for them. Each slot can be seen as having a relative value. A simple analysis of the sponsored search auction~\cite{edelman2005internet} treats the click-through rate (CTR) of an advertisement as correlated with the slot positioning and hence this is where the relative value for each slot emerges. 
}

\added[]{
Two famous auction mechanisms are the Vickrey-Clark-Groves auction (VCG)~\cite{ausubel2006lovely} and the Generalized Second Price auction (GSP)~\cite{edelman2005internet}. VCG is a powerful mechanism that is strategy proof and efficient (in every instance). In VCG the $i$th highest bidder gets the $i$th best good, but pays the \emph{externality} that she imposes on the other bidders by winning that slot. A mechanism like that is difficult to be understood by a normal advertiser (it is also difficult to determine its valuation), making VCG not such a suitable candidate for sponsored search. On the other hand there is GSP, were the $i$th highest bidder gets the $i$th best good, but pays to the seller the bid of the $(i+1)$st highest bidder. This mechanism is neither strategy proof neither efficient but is easily understandable by the advertiser. In fact this mechanism was developed by Google and it is also used by Yahoo!. Adwords~\cite{adworks2} of Google uses this mechanism but also takes into account a ``quality score"~\cite{adworks}. Part of that quality score is the click through rate  that a specific advertiser is estimated to obtain if she gets a specific slot. In this case, the allocation mechanism is basing its result on the product of the estimated CTR multiplied by the bid; and the payment of the bidder that wins a slot is obtained by multiplying her bid by the CTR. As mentioned in the previous paragraph we have modeled the CTR as a weight of the importance of each slot (like Edelman et al.~\cite{edelman2005internet}) and we assume that it is the same for all bidders .         
}

\added[]{
The key business ingredient of the current sponsored search model is promoting high quality advertisement (measured through CTR estimation) and in parallel increase of revenue for the search engine.
We argue that in addition to these features, fairness is a valid business model. Fairness, in the sense that no qualified advertiser would suffer publicity starvation (obtaining an ad slot few times or never). To this respect, \fedor could be used to guarantee fairness. To deal with the payments that the search engine would require, a flat fee for participating in $r$ rounds could be set. 
}

\added[]{
The mechanism could run centrally on the search engine and advertisers would communicate their desire to appear to a certain keyword search. 
For $r$ instances of this communication with the search engine the advertiser will pay a flat fee. \fedor guarantees that in expectation the advertiser will appear the same number of times in each slot. It also guarantees truthfulness and social efficiency, meaning that if the advertisers provide their true valuations on the slots they will be receiving a slot when they most value it. 
}

\added[]{
This business model could also apply to recommendation engines. Consider the example of Amazon's recommendation engine, that after a search for a good recommends to you several alternative goods. The way in which a decision is taken on which similar good should be presented to the customer can be based on the \fedor mechanism, so that fairness among all similar goods will be guaranteed. A similar example is the one of showing adds to a newspaper reader related to the content of the news she read. 
}

\added[]{
Comparing \fedor to the celebrated VCG mechanism, it is clear that VCG has an advantage since its properties hold deterministically. On the other hand, \fedor 's properties hold on expectation. This means that \fedor can only be applied on a repeated setting where the agents remain the same. Nevertheless, simulations have shown us that besides fairness, applying \fedor to a setting where the need of payments is inherent has another benefit. Due to the necessary flat participation fee scheme, the search engine (i.e., the seller) can set this fee in such a way to control her utility and the utility of the advertisers.  
}

%%%%%%%%%%%%%%%%%%%%%%%%%%%%%%%%%%%%%%%%%%%%%%%%%%%%%%%%%%%%%%%%%%%%%%
%%%%%%%%%%%%%%%%%%%%%%%%%%%%%%%%%%%%%%%%%%%%%%%%%%%%%%%%%%%%%%%%%%%%%%
%%%%%%%%%%%%%%%%%%%%%%%%%%%%%%%%%%%%%%%%%%%%%%%%%%%%%%%%%%%%%%%%%%%%%%
 
\paragraph{Previous and Related Work}
The concept of mechanisms without money has been studied before by Schummer and Vohra\cite{withoutmoney} and Procaccia and Tennenholtz\cite{DBLP:journals/teco/ProcacciaT13}. We are focusing on mechanisms for resource allocation/distribution. To this respect, a related work is the one of Guo et al. \cite{DBLP:conf/wine/GuoCR09}. They study the problem of allocating a single item over multiple competing agents in a repeated setting, without the involvement of money. To do this, they introduce an artificial payment system; they propose a number of repeated Bayes-Nash incentive compatible mechanisms and analyze how competitive they are with mechanisms using money. In their setting they assume, like we do, that agents learn their private values before each interaction, and also that preferences are i.i.d according to a distribution that does not change over time. In a later work, Guo and Conitzer\cite{DBLP:conf/atal/GuoC10} design a strategy proof mechanism for allocating multiple heterogeneous goods among two agents, in a single shot, prior-free and payments-free setting. \fedor promotes a truthful declaration of values, since telling the truth increases the agents benefit (assuming the possibility of the agents being more than two) independently of the other agents strategy. Our mechanism considers a setting of multi-round interactions among the agents, for allocating a set of heterogeneous goods to them. We assume that all agents distributions can be transformed into a uniform [0,1] distribution, but besides that we have no information about the distribution of the agents. So, \fedor is \deleted[]{a prior-free and} payment-free mechanism that achieves fairness and efficiency. 

\added[]{The work of Moscibroda and Schmid~\cite{schmidinfocom} investigates mechanisms without payments for throughput maximization and compares their social welfare with payment mechanisms. In comparison with our work, these mechanisms can be applied in a non-repeated setting, but it is not guaranteed that a feasible solution in terms of social welfare can be found. Moscibroda and Schmid shade some light on the degree up to which payments are inevitable and the potential benefit from the use of mechanism without payments. In contrast to the kind of mechanisms they consider, in \fedor there is no need of a trusted entity.  
Due to the infeasibility of using money incentives in many computer science oriented problems, new techniques developed to substitute monetary incentives, such as the {\it tit-for-tat}  in BitTorrent~\cite{schmidbittorrent,otherbittorrent} and money-burning mechanisms~\cite{moneyburning}. Unfortunately, these techniques can not be easily generalized to other problems. In particular, in~\cite{interdomain} interdomain routing is analyzed from a game theoretic perspective and incentive-compatible
distributed mechanisms without payments are designed in a repeated setting. 
}

Rahman et al.~\cite{rahman2010improving} considers the fairness constraint focusing on P2P systems, and propose an alternative approach to resource allocation that achieves fairness and efficiency on effort-based incentives, as opposed to contribution or output-based incentives. This work is somehow related to ours by the fact that we also consider that fairness is not proportional to the valuations of the agents, but in our work no incentives are necessary to achieve efficiency. In addition, we give analytical proofs of the properties of our mechanism, unlike the work of Rahman et al.~\cite{rahman2010improving}. 

Our work is inspired by the concept of linking mechanisms proposed by Jackson and Sonnenschein\cite{jackson2007overcoming}. They show that when a lot of independent copies of the same decision problem are linked together, then no incentive constrains are needed for agents to be truthful. The spectrum of players' responses to a probability distribution is known by considering a budget restriction. They show that a linking mechanism is valid when the players' possible decisions are distributed following discrete probabilities. 
This work is a natural extension of the works of Santos et al.~\cite{QPQ,QPQarxiv}, were a mechanism was designed for allocating a single task, that every participating agent desires to have accomplished, but no agent wants to execute. The designed mechanism is called QPQ and also uses the concept of linking mechanism. Two fundamental deferences among \fedor and QPQ is: (1) they are solving problems with opposite goals, (2) \fedor considers multiple good assignment in a single round while QPQ solves the problem of single task assignments.

\paragraph{Contributions}

In this paper we propose a mechanism called \fedor that applies to a multi-instance setting, for solving the problem of allocating a number of heterogeneous goods to multiple agents in  a round of interaction. The mechanism solves the problem without monetary incentives. Our contributions are as follows:
\begin{itemize}
\item We propose a distributed solution for our mechanism, one in which no central authority is needed. This does not mean that a centralized solution can not take place, on the contrary. 
\item We analytically show that \fedor guarantees \emph{fairness}. Meaning that every agent will receive the $i$th best good, for every $i$, the same number of times in expectation.
\item We prove that the mechanism is \emph{socially efficient}, in the sense that the social utility (the total valuation of the goods assigned) is maximized, subject to the fairness property. This means that agents will often receive the $i$th best good when they value it the most (over the multiple interactions setting).  
\item We show that being honest and announcing the real valuation of the goods at every round maximizes the expected utility (value of goods received) of a player. This means that the mechanism is \emph{truthful}. 
\item \added[]{We motivate our social mechanism with a number of examples and apply it to a non-trivial paradigm. Through these examples and a discussion provided at the end, we motivate the need for the existence of mechanisms without payments.}
\item \deleted[]{We motivate our social mechanism with a number of examples and apply it to a non-trivial paradigm, that of sponsored search. The peculiarity of this paradigm is that money exchange is inherent. We evaluate the benefits of alleviating the constraint of money by: } 
\added[]{ Finally, we simulate our mechanism and we are able to compare it in the example of sponsored search with two classical auction mechanisms. In particular:}
\added[]{
\begin{itemize}
\item We compare experimentally \fedor with  the  established auction mechanisms Vickrey-Clarke-Groves (VCG) and (Generalized Second Price) GSP, in the presence of only honest workers. 
In contrast to the utilities in VCG and GSP that are fixed; by appropriately setting a flat participation fee to the auction, \fedor can achieve any utility for the seller and the player, in the sponsored search example. This advantage of our mechanism comes from the fact that \fedor is a payments-free mechanism.  
\item  Experimental results go hand by hand with our analytical results asserting the truthfulness and social efficiency properties. Moreover, experiments compare the performance in terms of utility between \fedor, VCG, and GSP; 
%\anto{showing that with \fedor the utility of honest players and the social utility are much larger 
%than those of VCG and GSP.} 
showing that with \fedor the utility of the honest players is independent of the percentage of cheaters or their cheating behavior, while in  VCG and GSP is dependent.
\end{itemize}
}
\end{itemize}

The rest of paper is organized as follows. In Section~\nameref{analysis} we formally present our model and notation used throughout the document. In the same section we present the \fedor mechanism and an analysis formally proving its properties. \deleted[]{In  Section~\nameref{s:application} we present the scenario of sponsored search as a non-trivial plausible application scenario for \fedor. More over we present simulation results showing graphically the properties of the mechanism.} \added[]{In Section~\nameref{lb:experiments}, we presented our simulation results comparing \fedor with two classical auction mechanisms.}  Finally in Section~\nameref{discussion} we discuss \added[]{the advantages and limitations of mechanisms without payments. Also we review}
 briefly the possible limitation of the mechanism and the future directions we are considering. 

\newpage

\newcommand{\D}{\ensuremath{\cal D}\xspace}
\section*{Analysis}
\label{analysis}
\subsection*{Model}
\label{s:model}

Before moving any further, let us define the setting we are considering as well as the problem we are solving. 

\begin{defn}[Setting]
\label{setting}
 We consider the presence of $n$ risk-neutral players, $N = \{1,2,\ldots, n\}$. Every one of the $n$ players participates in a (a priori unknown)  number of consecutive instances (each instance is considered a round) of an allocation game. In each instance of the allocation game: 
\begin{enumerate}
\item $k$ heterogeneous non-divisible goods are allocated among the $n$ players,  where $1 \leq k<n$.
\item For every good $i$, of the $k$ goods of the round, there exist a relative value depicted by the weight $w_i$ (goods in each round can be different but the relative values remain the same). We assume that goods are sorted by decreasing weight, i.e., $w_1\geq w_2\geq \cdots \geq w_k$.
\item Players declare how much they value obtaining any of the goods of the set, this is expressed by a single value. 
\end{enumerate} 
\end{defn}

The outcome set of the game played is the set $\D$ of all $k$-permutations of $N$ 
(i.e., all posible permutations of the subsets of $N$ of size $k$).
Hence, the outcome $d=(d_1, d_2, \ldots, d_k) \in \D$ is the ordered list of players to 
whom the slots will be allocated, so that the player $d_i$ will receive the
$i$th slot (whose weight is $w_i$).

Each player \emph{privately} observes her preferences 
over the alternatives in $\D$ before the collective choice in the game is made. 
This is modeled by assuming that player $i$ privately observes a parameter $\theta_i$, which determines her preferences over obtaining a good from the set. 
We say that $\theta_i$  is the player type, for every player $i$. The set of all possible types of player $i$ is $\Theta_i$.
We denote by $\theta = (\theta_1 , \theta_2 , \dots , \theta_n )$ the vector of player types. The set of all possible vectors is $\Theta = \Theta_1 \times \Theta_2  \times \dots \Theta_n$. We denote by $\theta_{-i}$ the vector obtained by removing $\theta_i$ from $\theta$. 

Thus, we denote by $\Pi = \Delta(\Theta)$ (in general, we denote by $\Delta(S)$ the set of all probability distributions over some set $S$). 
the set of all probability distributions over $\Theta$. It is assumed that there is a common prior distribution $\pi \in \Pi$ that is shared by all the players. 
We denote by $\pi_i \in \Delta(\Theta_i)$ the marginal probability of $\theta_i$. 
We define $\pi_i (\theta_i| \theta_{-i})$ as the conditional probability distribution of $\theta_i$ given $\theta_{-i}$. 

In this work we assume that the player types are normalized, in the sense that $\Theta_i=[0,1]$, for every player $i$. Moreover, for every player $i$ the marginal probability $\pi_i$ of $\theta_i$ is also normalized as the uniform continuous distribution in the interval $[0,1]$. This are not simple assumptions as in the work of Santos et al.~\cite{QPQ,QPQarxiv}, applying the Probability Integral Transformation (PIT) we can transform any probability distribution to the uniform [0,1] Finally, we assume that all distributions $\pi_i$ are independent, i.e., $\pi(\cdot|\theta_{-i})=\pi_i(\cdot)$. For this reasons we have made the simplification of having our notations free from the concept of the round.

Players' preferences over outcomes are represented by a utility function $u_i (d, \theta_i ) \in \mathbb{R}$ defined over all $d \in \D$ and $\theta_i \in \Theta_i$. In this paper the utility function is defined as $u_i(d, \theta_i )=\theta_i \sum_{j=1}^k w_j \delta_{i d_j}$, $\forall i \in N, \forall d \in \D$, where $\delta_{ab}$ is the Kronecker delta.

We assume that the set of outcomes $\D$, the set of players $N$, the type sets in $\Theta$, the
common prior distribution $\pi \in \Pi$, and the utility functions $u_i$ are \emph{common knowledge} for all the players. 
Similarly, the game rules defined by the mechanism used are also common knowledge. 
However, the specific type $\theta_i$ observed by player $i$ belongs to her \emph{private information}. 

A strategy for the player $i$ is a map $\sigma_i: \Theta_i \rightarrow \Delta(\Theta_i)$, where $\sigma_i(\hat{\theta}_i|\theta_i)$ is the conditional probability that the player reports $\hat{\theta}_i$ when her true type is $\theta_i$. 
A strategy $\sigma_i$ is \textit{truthful} (and we say that the player is \emph{honest}) if, for every pair $(\hat{\theta}_i, \theta_i)$,  $\sigma_i(\hat{\theta}_i|\theta_i) = 1$ if $\hat{\theta}_i = \theta_i$ and $0$ otherwise.
%A \emph{finite deviation} (or simply deviation) is any reporting strategy that is not \textit{truthful}, i.e., the player lies with positive probability conditional on some type. 
As usually done, we will use $\hat{\theta}_i$ to denote the reported type, and $\theta_i$ the actual type. 
%We use $\sigma(\hat{\theta}|\theta)$ to denote the aggregation of the strategies $\sigma_i$ of all players in $N$.

In the paper we assume the availability of a \emph{goodness of fit} (GoF) test, that can test whether a value (a reported type in our case)
is a sample from a uniform distribution in $[0,1]$. Hence, $\mathit{GoF\_Test}(\hat{\theta}_i)$ is true if and only if $\hat{\theta}_i$ is a uniform random sample. By this definition we assume that the GoF test is perfect, but this is a theoretical artifice. \deleted[]{In Section~\ref{s:application} we give a good approximation to a perfect GoF test. } \added[]{In Section~\nameref{lb:experiments}, through simulations we find   a good approximation to a perfect GoF test.}

In this paper, we search for a mechanism $\langle \Theta, g \rangle$, where $g: \hat{\Theta} \rightarrow \Delta(\D)$ is the decision function. With this function, $g(d|\hat{\theta})$ is the conditional probability that the mechanism decides $d \in \D$ when the players report $\hat{\theta}$. The mechanism must be
without utility transfers (payments) and satisfy the following properties.
\begin{itemize}
\item 
\emph{Fairness.} Every player $i$ (honest or not) gets the $j$th slot, the same proportion of times in expectation, for every $j$. I.e.,
%$Pr[i=g(\hat{\theta})_j]=Pr[i'=g(\hat{\theta})_j], \forall i,i' \in N, \forall j \in [1,k]$.
\begin{equation} \E_{\theta, \hat\theta \in \Theta}\left[\sum_{d \in D: d_j=i} g(d|\hat{\theta})\right] = \E_{\theta, \hat\theta\in \Theta}\left[\sum_{d \in D: d_j=i'} g(d|\hat{\theta})\right], \forall i,i' \in N, \forall j \in [1,k], \forall \sigma.\end{equation}

\item
\emph{Truthfulness.} The strategy that maximizes the utility of a player $i$ is to be honest. I.e., 
%$E[u_i(g(\theta_i,\hat{\theta}_{-i}), \theta_i)] \geq E[u_i(g(\hat{\theta}), \theta_i)], \forall i \in N$.
\begin{equation} \E_{\theta, \hat\theta\in \Theta}\left[ \sum_{d \in \D} u_i(d, \theta_i) g(d | \theta_i,\hat{\theta}_{-i}) \right] \geq \E_{\theta, \hat\theta\in \Theta} \left[\sum_{d \in \D} u_i(d, \theta_i) g(d | \hat{\theta}) \right], \forall i \in N, \forall \sigma. \end{equation}

\item
\emph{Social Efficiency.} If all players are honest, the expected social utility with decision function $g$ is maximized with respect to any other decision function $g'$. I.e.,
\begin{equation}
\E_{\theta\in \Theta}\left[ \sum_{i \in N} \sum_{d \in \D} u_i(d, \theta_i) g(d | \theta) \right] \geq \E_{\theta\in \Theta}\left[ \sum_{i \in N} \sum_{d \in \D} u_i(d, \theta_i) g'(d | \theta) \right].
\end{equation}

\end{itemize}

\subsection*{The \fedor Mechanism}
\label{lb:mechanism}
%

%\vspace{-0.5em}
%\subsection{Simple Ideal \fedor Mechanism}

%We use $\overline{b}_i$ to denote the PIT-normalized bid to be published, while $b_i$ is the actual bid. We also express by $\overline{b}_i$ the pseudorandom value that replaces the value published by $i$ when it does not pass the acceptance test. 

\fedor solves the problem of allocating multiple resources in a multi-instance setting, with a decision function that is such that the properties of fairness,  truthfulness and social efficiency are satisfied without utility transfers.  Algorithm~\ref{alg3} describes the steps that a generic player $i$ executes and the actions taken by the mechanism. This algorithm represents how the mechanism can be applied in a distributed setting, in the absence of a central authority that takes the allocation decision.   

In every instance of the allocation game (defined as a new announcement of $k$ goods to be allocated) a player $i$ observes its type $\theta_i$. Then, applying its strategy $\sigma_i$, the player chooses a
type $\hat{\theta}_i$ that will be reported as the player's bid. These bids are gossiped among all players, so that all of them end with the same vector of reported types $\hat{\theta}$. We are assuming that communication among players is perfect (reliable and synchronous). Once a player has received all the reported types from the rest of the players, the algorithm moves on, applying the mechanism. 

Each player applies the mechanism to obtain the allocation decision. The decision function $g_f:\Theta \rightarrow D$ of \fedor is one where $g_f(\hat{\theta})=(i_1,i_2, \ldots i_k)$ such that $(i_1,i_2, \ldots i_k):v_{i_1}\geq v_{i_2} \geq \ldots v_{i_k}$ and $v_{i_k} \geq v_i \forall i \in N\setminus\{i_1,..i_k\}$ where
$$
v_i = \begin{cases} \hat{\theta_i} & \quad \mathrm{GoF test}(\hat{\theta_i})\\
\mathrm{pseudorandom}(\hat{\theta}_{-i}) & \quad \mathrm{otherwise}\\ \end{cases} 
$$  
In other words the decision function of \fedor is one that given the vector of declared types $\hat{\theta}_i$ provides the vector of players $(i_1,i_2, \ldots i_k)$ receiving the $k$ goods, by $i_j$ we declare the player receiving the $j$th good. The decision of the mechanism is based on a cheaters (following a strategy besides declaring their true type) detection scheme. Recall that, as we explained in Subsection~\ref{s:model} each player has a uniform continuous distribution in the interval [0,1] over all instances. This is a common information that all players share and allows them to verify announced types of the rest of the players. If the player $i$ declared type $\hat{\theta}_i$ passes the goodness of fit (GoF) test then the declared value is accepted. Otherwise the declared value is replaced, by each player, by a pseudorandom value uniform in $[0,1]$ derived from the rest of reported types $\hat{\theta}_{-i}$. Since all players apply the same perfect goodness-of-fit test to the same
value $\hat{\theta}_i$, and the same pseudorandom function to the same vector $\hat{\theta}_{-i}$, they all assign the same value to $v_i$. As a result, the same vector $v$ is obtained
by each player, 
which contains the values of the players that passed the test and the generated values for those that did not.        
The mechanism allocates the $k$ goods in order to the $k$ players with the highest value in the vector $v$.     

It is important to notice that the algorithm is the same for all players and that it is based on information known by all of them. Therefore, no central entity is required to apply the
\fedor mechanism. Of course, if the environment is such that a central authority is available, the algorithm can trivially be transformed to accommodate a centralized solution using \fedor. 

\subsection*{Formal Analysis}
\label{s:analysis}

In this section we analyze the \fedor mechanism and we show that the three desired properties described in Subsection~\nameref{s:model}, namely fairness, truthfulness, and social efficiency are satisfied. We start by proving the latter. 

We show that \fedor is socially efficient if all players are honest. To prove this, first we prove that there is no mechanism that has a decision function that gives a higher utility to a player, independently of the declared types  $\Theta$.    
\begin{lemma}
\label{optimality}
There is no mechanism $\langle \Theta, g \rangle$ such that 
$$E[\sum_{i \in N} \sum_{d \in \D} u_i(d, \theta_i) g(d | \theta)] > E[\sum_{i \in N} \sum_{d \in \D} u_i(d, \theta_i) g_f(d | \theta)].$$
\end{lemma}

\begin{proof}
Assume the claim is not true. Then, there is a mechanism $\langle \Theta, g \rangle$ such that 
\begin{equation}
E[\sum_{i \in N} \sum_{d \in \D} u_i(d, \theta_i) g(d | \theta)] > E[\sum_{i \in N} \sum_{d \in \D} u_i(d, \theta_i) g_f(d | \theta)].
\end{equation}
This means that there is at least one $\theta$ such that 
\begin{equation}
\sum_{i \in N} \sum_{d \in \D} u_i(d, \theta_i) g(d | \theta) >  \sum_{i \in N} \sum_{d \in \D} u_i(d, \theta_i) g_f(d | \theta).
\end{equation}

Let $d_x=(x_1,x_2, \ldots, x_k)$ be an outcome that maximizes  $\sum_{i \in N} u_i(d_x, \theta_i)$. \\* Then, 
$\sum_{i \in N} u_i(d_x, \theta_i)  \geq \sum_{i \in N} \sum_{d \in \D} u_i(d, \theta_i) g(d | \theta)$.
On the other hand, let $g_f(d_f | \theta)=1$ for $d_f=(i_1, i_2, \ldots,i_k)$. By the definition of $u_i$, this means that
$\sum_{j=1}^k \theta_{x_j} w_j > \sum_{j=1}^k \theta_{i_j} w_j$. However, this is not possible, since $w_1 \geq w_2 \geq \cdots \geq w_k$, and the values $\theta_{i_j}$ are the largest values
in $\theta$ in decreasing order, i.e., $\theta_{i_1} \geq \theta_{i_2} \geq \cdots \geq \theta_{i_k}$. Hence, the mechanism $\langle \Theta, g \rangle$ does not exist.
\end{proof}

Since there is no mechanism that has a decision function that is better than the one of \fedor mechanism, then we can also proof that when all players are honest (a.k.a declaring their true type) the social utility (i.e. the total utility of all players) is maximized.  
\begin{theorem}[Social efficiency]
\label{optimality2}
For any mechanism $\langle \Theta, g \rangle$ and strategies $\sigma_i, \forall i \in N$, it holds that 
$$E[\sum_{i \in N} \sum_{d \in \D} u_i(d, \theta_i) g_f(d | \theta)] \geq E[\sum_{i \in N} \sum_{d \in \D} u_i(d, \theta_i) g(d | \hat{\theta}) | \sigma_1, \ldots, \sigma_n].$$
\end{theorem}

\begin{proof}
Assume the claim is not true, and hence there is a mechanism $\langle \Theta, g \rangle$ and strategies $\sigma_i, \forall i \in N$, such that 
\begin{equation}
E[\sum_{i \in N} \sum_{d \in \D} u_i(d, \theta_i) g_f(d | \theta)] < E[\sum_{i \in N} \sum_{d \in \D} u_i(d, \theta_i) g(d | \hat{\theta}) | \sigma_1, \ldots, \sigma_n].
\end{equation}
%$E[\sum_{i \in N} u_i(g_f(\theta), \theta_i)] < E[\sum_{i \in N} u_i(g(\hat{\theta}), \theta_i)| \sigma_i, \forall i \in N]$.
Then, let us define a new mechanism $\langle \Theta, g' \rangle$ as follows.
{\small$
g'(d | \theta) = \int_0^1\int_0^1\ldots\int_0^1 g(d | \hat{\theta})\cdot \sigma_1(\hat{\theta}_1 | \theta_1) \cdot \ldots \cdot \sigma_n(\hat{\theta}_n | \theta_n) d \hat{\theta}_n d \hat{\theta}_{n-1} \ldots d \hat{\theta}_1.
$}
The decision function $g'(\cdot)$ assigns the same probability to each possible outcome $d$ as the combination of the strategies $\sigma_i(\cdot)$ and the decision function $g(\cdot)$. 
Thus, 
{\small
\begin{eqnarray}
E[\sum_{i \in N} \sum_{d \in \D} u_i(d, \theta_i) g_f(d | \theta)]
<  E[\sum_{i \in N} \sum_{d \in \D} u_i(d, \theta_i) g(d | \hat{\theta}) | \sigma_1, \ldots, \sigma_n] \\
=  E[\sum_{i \in N} \sum_{d \in \D} u_i(d, \theta_i) g'(d | \theta)],
\end{eqnarray}}
contradicting Lemma~\ref{optimality}.     
%
%
%Assume there exist a mechanism $M$ with function $f_M$ and a strategy $f_i, \forall i$ such that $E[\sum_{i=1}^n \hat{U}^M_i|f_i,  \forall i]>E[\sum_{i=1}^n \overline{U}_i]$.
%Define $M'$ as $f_{M'}(D,x_1,x_2,\ldots,x_n)= \int_0^1\int_0^1\ldots\int_0^1f_M(D,z_1,z_2,\ldots,z_n)\cdot f_1(z_1,x_1)\cdot \ldots \cdot f_n(z_n,x_n)dz_ndz_{n-1}\ldots dz_1$ thus $E[\sum_{i=1}^n \overline{U}^{M'}_i]=E[\sum_{i=1}^n \hat{U}^{M}_i|f_i,  \forall i]>E[\sum_{i=1}^n \overline{U}_i]$ contradicting Lemma~\ref{optimality}.   
\end{proof} 

We will now start proving the truthfulness property of the \fedor mechanism. Let us denote by $E[\hat{U}_i | \sigma_1, \ldots,\sigma_n]$ the expected utility of player
$i$ under the \fedor mechanism when players follow strategy $\sigma_i, \forall i \in N$.

\begin{lemma}
\label{final2}
For any non empty set of players $S \subset N$, the total expected utility of the players in $S$ does not depend on the strategies of the rest of players. 
\end{lemma}

\begin{proof}
We start by proving that the expected utility of a player $i$ does not depend on the strategy of another player $h \neq i$. In particular, we will show that 
$E[\hat{U}_i | \sigma_x, \forall x \in N] = E[\hat{U}_i | \sigma_x, \forall x \neq h; h \ \mathrm{honest}]$, i.e.,
%$E[\sum_{d \in \D} u_i(d, \theta_i) g(d | \hat{\theta}) | \sigma_j, \forall j] = E[\sum_{d \in \D} u_i(d, \theta_i) g(d | \hat{\theta}, \theta_h) | \sigma_j, \forall j \neq h]$,
that the utility of $i$ does not depend on whether $h$ is honest or not, which implies the former statement.

For every player $x \in N$, let $f_x(v_x,\theta_x)$ be the density function of the pairs of values $(v_x,\theta_x)$, where $\theta_x$ is the type observed by player $x$ and $v_x$ is the
entry corresponding to $x$ in the vector $v$ used to decide in \fedor (see Algorithm~\ref{alg3}). If the strategy $\sigma_x(\cdot)$
induces a uniform distribution on the reported types $\hat{\theta}_x$ (recall that $\theta_x$ does follow a uniform distribution), then
the values $v_x$ will always be the declared types (bid) $\hat{\theta}_x$. Otherwise, the GoF test will always fail and $v_x$ will be pseudorandom values generated
from the rest of declared values. In any case, the values $v_x$ are uniformly distributed, and the function $f_x$ is independent from the rest of functions $f_i, \forall i\neq x$.
Moreover, the marginal distributions of $f_x(v_x,\theta_x)$ must also be uniform.
% and uniformly distributed with respect to the value $\theta_j$.

Let us consider any player $i \in S$. The expected utility $E[\hat{U}_i | \sigma_x,$ $\forall x \in N]$
of the player can then be computed as
\begin{eqnarray}
& & \sum_{d \in \D} \underbrace{\int_0^1 \cdots \int_0^1}_{2n} u_i(d | \theta_i) g_f(d | v) f_1(v_1,\theta_1)\cdots f_n(v_n,\theta_n) \\
& & d v_1 \cdots dv_n d \theta_1 \cdots d \theta_n \\
& = & \sum_{d \in \D} \underbrace{\int_0^1 \cdots  \int_0^1}_{2n-1} u_i(d | \theta_i) g_f(d | v) f_1(v_1,\theta_1)\cdots \\
& & f_{h-1}(v_{h-1},\theta_{h-1}) f_{h+1}(v_{h+1},\theta_{h+1}) \cdots f_n(v_n,\theta_n) \\
& & \int_0^1 f_h(v_h,\theta_h) d \theta_h d v_1 \cdots dv_n d \theta_1 \cdots d \theta_{h-1} d \theta_{h+1} \cdots d \theta_n \\
& = & \sum_{d \in \D} \underbrace{\int_0^1 \cdots  \int_0^1}_{2n-1} u_i(d | \theta_i) g_f(d | v) f_1(v_1,\theta_1)\cdots \\
& & f_{h-1}(v_{h-1},\theta_{h-1}) f_{h+1}(v_{h+1},\theta_{h+1}) \cdots f_n(v_n,\theta_n)  \\
& & d v_1 \cdots dv_n d \theta_1 \cdots d \theta_{h-1} d \theta_{h+1} \cdots d \theta_n.
\end{eqnarray}

The first equality follows from the independence of $f_h$ from the other $f_j$. The last equality follows from the fact that $\int_0^1 f_h(v_h,\theta_h) d \theta_h = 1$, from the uniform marginal distribution of $f_h$ with respect to $\theta_h$. Now, applying a change of variable, we can replace in the above expression $v_h$ by $\theta_h$, resulting after reordering in
\begin{eqnarray}
E[\hat{U}_i  | \sigma_x, \forall x \in N]  =   {\tiny\sum_{d \in \D} \underbrace{\int_0^1 \cdots  \int_0^1}_{2n-1}} u_i(d | \theta_i) g_f(d | v_{-h}, \theta_h) \\
f_1(v_1,\theta_1)\cdots  f_{h-1}(v_{h-1},\theta_{h-1})  f_{h+1}(v_{h+1},\theta_{h+1}) \cdots \\
f_n(v_n,\theta_n) d v_1 \cdots d v_{h-1} d \theta_h d v_{h+1} \cdots dv_n d \theta_1 \cdots  d \theta_n \\
 =  E[\hat{U}_i | \sigma_x, \forall x \neq h; h \ \mathrm{honest}].
\end{eqnarray}

This property can now be applied iteratively for every player $h \notin S$, showing that the expected utility of $i \in S$ is independent of the strategy of the
players in $S$. Summing over all the players in $S$ the claim of the lemma is proved.
\end{proof}

Using Lemma~\ref{final2} we can proof the truthfulness property. 

\begin{theorem}[Truthfulness]
\label{thetheorem}
The strategy that maximizes the utility of a player $i$ is to be honest. I.e., 
$E[\hat{U}_i | \sigma_x, \forall x \in N ] \leq E[\hat{U}_i | \sigma_x, \forall x \neq i; i \ \mathrm{honest} ]$.
\end{theorem}

\begin{proof}
Assume for contradiction that the claim is not true. Then, there are strategies $\sigma_x, \forall x \in N$ such that
\begin{equation}
E[\hat{U}_i | \sigma_x, \forall x \in N ] > E[\hat{U}_i | \sigma_x, \forall x \neq i; i \ \mathrm{honest} ]. \label{eq1}
\end{equation}
From Lemma~\ref{final2} we have the following identities (the third one uses linearity of expectations).
\begin{eqnarray}
\!\!\!\!\!\!\!\!E[\hat{U}_i | \sigma_x, \forall x \in N ]  =  E[\hat{U}_i | \sigma_i; x \text{ is honest}, \forall x\neq i ] \label{eq2} \\
\!\!\!\!\!\!\!\!E[\hat{U}_i | \sigma_x, \forall x \neq i; i \ \mathrm{honest} ] =  E[\hat{U}_i | \text{all players are honest}] \label{eq3} \\
\!\!\!\!\!\!\!\!E[ \sum_{h\neq i} \hat{U}_h | \sigma_i; x \text{ is honest}, \forall x\neq i ]  =  E[\sum_{h\neq i} \hat{U}_h | \text{all players are honest}] \label{eq4}
\end{eqnarray}
Replacing equations \ref{eq2} and \ref{eq3} in \ref{eq1}, we obtain
\begin{equation}
E[\hat{U}_i | \sigma_i; x \text{ is honest}, \forall x\neq i ] > E[\hat{U}_i | \text{all players are honest}].
\end{equation}
Which added to \ref{eq4} yields
\begin{equation}
E[ \sum_{h\in N} \hat{U}_h | \sigma_i; x \text{ is honest}, \forall x\neq i ] > E[\sum_{h\in N} \hat{U}_h | \text{all players are honest}].
\end{equation}
However, this contradicts Theorem~\ref{optimality2}.
\end{proof}

We concentrate now in the fairness property. Recall the \fedor mechanism as it was described in the previous subsection. Either because the player reports types that are uniformly distributed, or because they are replaced by \fedor with uniform pseudorandom values. The values $v_i$ of the vector $v$ used to distribute the goods are independent random samples form a uniform distribution $[0,1]$. Hence the following theorem.

\begin{theorem}[Fairness]
\label{fairness-goods}
Every player $i$ (honest or not) has the same probability $1/n$ of getting the $j$th slot, for every $j$.
\end{theorem}

\begin{proof}
Since the elements of the vector $v$ are independent uniform samples as mentioned, the probability that the value $v_i$ is the $j$th largest in $v$ is
$\left( {\begin{array}{*{20}c} n-1 \\ j-1 \\ \end{array}} \right) (1-v_i)^{j-1}v_i^{n-j}$. Then the probability that a player $i$ has the $j$th largest value in $v$, 
given that $v_i$ takes values uniformly between 0 and 1 is
\begin{equation}
\int_0^1 \left( {\begin{array}{*{20}c} n-1 \\ j-1 \\ \end{array}} \right) (1-v_i)^{j-1}v_i^{n-j}d v_i
=\frac{1}{n}
\end{equation}
\end{proof}

Finally, we will analyze the expected utility of the players depending on their strategy. We start by presenting the expect utility of an honest player. Observe that the utility obtained is independent from the strategies of the rest of players, as proven in Lemma~\ref{final2}.
%In the utility values obtained it is not included the flat fee that we assume will be paid by the player to participate in the auctions. This fee can be simply subtracted from the utility presented here to obtain the actual utility. We start by presenting the expect utility of an honest player. Observe that the utility obtained is independent from the strategies of the rest of players.

\begin{theorem}
\label{utility-honest}
The expected utility of an honest player $i$ is \\* $E[\hat{U}_i | \sigma_x, \forall x \neq i; i \ \mathrm{honest} ] = \frac{\sum_{j=1}^k w_j(n-j+1)}{n(n+1)}$. This value is independent from
the strategies $\sigma_x, \forall x \neq i$.
\end{theorem}

\begin{proof}
Since $i$ is honest, she reports her true type $\theta_i$, which follows a uniform distribution. Hence, it holds that $v_i=\theta_i$. As in the proof of Theorem~\ref{fairness-goods}, the probability that the type $\theta_i$ is the $j$th largest value is $\left( {\begin{array}{*{20}c} n-1 \\ j-1 \\ \end{array}} \right) (1-\theta_i)^{j-1}\theta_i^{n-j}$. Then, the expected utility can be computed as
\begin{eqnarray}
\begin{aligned}
E[\hat{U}_i | \sigma_x, \forall x \neq i; i \ \mathrm{honest} ] \\
 =  \sum_{j=1}^k w_j \int_0^1 \theta_i \left( {\begin{array}{*{20}c} n-1 \\ j-1 \\ \end{array}} \right) (1-\theta_i)^{j-1}\theta_i^{n-j}d \theta_i \\
=  \frac{\sum_{i=1}^j w_j(n-j+1)}{n(n+1)}
\end{aligned}
\end{eqnarray}

\end{proof}

We now find the utility of a player $i$ that is not honest (a.k.a cheater), but she reports values that are not uniform (and hence the goodness-of-fit test always fails) 
or she reports types $\hat{\theta}_i$ that are uniform but independent of her true types $\theta_i$. 
Observe that this theorem does not consider the case when the types reported are uniform and somehow correlated with the real types (this is left to be shown experimentally).

\begin{theorem}
\label{utility-dishonest}
The expected utility of a dishonest player $i$ that reports non uniform types or types independent of her true normalized uniform distribution is 
$E[\hat{U}_i | \sigma_x, \forall x \in N ] = \frac{1}{2n} \sum_{j=1}^k w_j$.
\end{theorem}

\begin{proof}
In this case, the value $v_i$ used to distribute the slots follow a uniform distribution that is independent of the actual type $\theta_i$ of player $i$. Hence, 
\begin{eqnarray}
\begin{aligned}
E[\hat{U}_i | \sigma_x, \forall x \in N ] \\
= \sum_{j=1}^k w_j \int_0^1 \theta_i \int_0^1\left( {\begin{array}{*{20}c} n-1 \\ j-1 \\ \end{array}} \right) (1-v_i)^{j-1}v_i^{n-j}d v_i d \theta_i \\
= \frac{1}{2n} \sum_{j=1}^k w_j.
\end{aligned}
\end{eqnarray}

\end{proof}

\section*{Simulation Results}
\label{lb:experiments}

\added[]{\fedor is a mechanism for  distributing/allocating goods, without monetary incentives, in the presence of strategic agents. It guarantees an allocation of the goods that will be fair and socially efficient, given that the mechanism is truthful (strategy-proof). The good properties of the mechanism are feasible under a specific setting, discussed in Subsection~\nameref{s:model}.
Under that setting, we simulate \fedor in the scenario of sponsored search (presented in the Introduction section). Simulations compare \fedor with the VCG and GSP mechanisms in two cases: (1) only honest players are present, (2) dishonest players are present. Additionally, simulations examine the cost of not having a perfect GoF test and the length of the historical test that would yield an almost perfect GoF test.} 

\added[]{
\paragraph*{Presence of only honest players:} 
In this scenario we look at an instance where the players are following an honest strategy. This could happen either because the mechanism is strategy proof, which is the case of VCG and \fedor or because they just choose to follow an honest strategy. We consider the scenario of sponsored search and we compute the utilities of the seller and and the players in all three mechanisms. The utilities are computed in two distinct scenarios: (a) the presence of 9 players is assumed and the number of the available ad slots $k$ varies from $1$ to $8$; the weight for each slot values in decreasing order from $k$ to $1$. (b) the number of ad slots is fixed to $k=3$ and the number of players varies from $4$ to $9$. Additionally we assume that, for the purpose of experimentation, the GoF test used is perfect. 
The results obtained are presented in Fig~\ref{comparison}, where we plot the utility achieved by one of the players (all players are honest and follow a uniform distribution) against the utility of the seller. In the case of \fedor, we plot the utility values assuming that the same flat fee is paid by every player,
so that the utility of the seller is simply this flat fee multiplied by the number of players. The utility of a player, on the other hand, is the value she assigns to the slots
she gets (which is essentially the value given in Theorem~\ref{utility-honest}) minus the flat fee.
}

\added[]{
In Fig~\ref{comparison}(a) it can be observed that the utilities are the same for VCG (with externality) and GSP in the case of auctioning one slot. 
%It is clear that VCG and GSP are performing the same in the sense of utility in the case of one slot. 
When the number of slots increases, the utility of the seller is greater with GSP than with VCG, while on the other hand the utility of the player 
is greater with VCG than with GSP. This behavior was expected since GSP is
a mechanism that favours the seller. Moreover, it is not strategy proof, something that would not put the player in an advantageous position if she is honest compared to VCG. 
 The utilities with the \fedor mechanism, on its hand, form a line that is a function of the flat fee players pay.
It is worth to observe that all the points that correspond to GSP and VCG utilities are on the \fedor line, which means that with the appropriate
value of the flat fee, \fedor can achieve the same seller and player utililites as GSP and VCG. Moreover,
the advantage of \fedor, as it is shown in Fig~\ref{comparison}(a), is that the flat fee provides a 
%can be fixed in such a way that the mechanism will have the desired 
tradeoff between seller and player utilities, and allows to chose any point in the lines shown in the figure. Hence, \fedor has an adjustable performance that can be changed depending on the needs of the seller. In Fig~\ref{comparison}(b) we make similar observations and derive the same conclusions. Again the utilities of
GSP and VCG lie on the \fedor line, and the utilities with \fedor can be tuned with the flat fee.
In addition, we notice that when the number of players increases, their utility decreases in all three mechanisms. On the contrary, the utility of the seller increases in all three mechanisms when the number of player increases. 
}

\paragraph*{Presence of dishonest players: }
\deleted[]{In this section we present general results we extracted from simulating \fedor, VCG and GSP. We present different scenarios where \emph{cheating} (i.e., dishonest) players exist (while in the graphs presented in the previous section player were assumed to be honest). (The modeling of cheating behavior has been done with different
probability distributions.) Additionally, we no longer assume that the GoF test used is perfect, and evaluate the impact of this fact.   Under these assumptions, we show experimentally that \fedor is truthful and socially efficient. We also compare its performance in terms of utility with the other two mechanism, GSP and VCG.}
%
%\added[]{
%This section evaluates, experimentally, the social efficiency of \fedor compared to 	VCG and GSP in the presence of dishonest players. Additionally, it is also shown experimentally that \fedor is a truthful mechanism. We no longer assume that the GoF test used is perfect, and we evaluate the impact of this. Simulations first examine the conditions under which the selected GoF test performs optimally. After having defined the impact of the GoF test parameters, we compare the utility of \fedor with VCG and GSP in a number of scenarios where the cheating behavior of the players varies. Having as a base case the scenario where all players are honest, we compare the way in which the honest and dishonest players' utility is affected. Finally, we show how the social utility is affected by the presence of dishonest players in all three mechanisms considered.     
%}
\added[]{In this section we assume that the players are dishonest set as our GoF test the Kolmogorov–Smirnov (KS) test and we evaluate its performance. Through simulations we examine the conditions under which the KS test performs optimally and we evaluate the impact of this optimality to our mechanism. Simulations show that  \fedor is a truthful mechanism. We compare the utility of \fedor with VCG and GSP in a number of scenarios where the cheating behavior of the players varies. Having as a base case the scenario where all players are honest, we compare the way in which the honest and dishonest players' utility is affected. Finally, we show how the social utility is affected by the presence of dishonest players in all three mechanisms considered.     
}

Unless otherwise stated, our results have been obtained by running $100$ experiments of $10000$ rounds (auctions) each. 
The number of participating players \deleted[]{has been nine and the number of 
goods/slots assumed has been $k=3$. The weights signed to the goods has been $w_1=3, w_2=2, w_3=1$.} \added[]{is $9$ and the number of slots is $k=3$. The weights assigned to the slots are $w_1=3, w_2=2, w_3=1$.}  Recall that honest players are considered the ones that reveal their true valuation for the set of goods. Cheater players are the ones declaring a different valuation from their actual valuation on the set of goods. All players' private valuations follow independent uniform distributions in the interval $[0,1]$.  

\added[]{To evaluate the impact of the presence of dishonest players in their utility, the utility of the honest players, and also in the social utility, we have simulated $10$ distinct scenarios. The scenarios we consider are as follows. Scenario $A$: $9$ honest players following the uniform distribution, $B$: $8$ honest players and one cheater with normal distribution ($\mu=0.5, \sigma=0.15$), $C$: $8$ honest players and one cheater with beta distribution with $\beta=0.9$, $D$: $8$ honest players and one cheater with beta distribution with $\beta=0.7$, $E$: $8$ honest players and one cheater with a random uniform distribution (different from its own distribution of values),  $F$: $6$ honest players and $3$ cheaters with random uniform distributions, 
$G$: $6$ honest players and $3$ cheaters with beta distributions ($\beta=0.9$), $H$: $6$ honest players and $3$ cheaters with beta distributions ($\beta=0.7$), $I$: $6$ honest players and $3$ cheaters with normal distributions ($\mu=0.5, \sigma=0.15$), $J$: $5$ honest players, $1$ cheater with random uniform distribution, $1$ cheater with beta distribution ($\beta=0.9$), $1$ cheater with beta distribution ($\beta=0.7$), and $1$ cheater with normal distribution ($\mu=0.5, \sigma=0.15$).}

%\begin{figure*}[t!]
%\begin{center}$
%\begin{array}{ccc}
%\hspace{-3em}
%\includegraphics[width=1.9in, trim = 0.2mm 0mm 30mm 2mm, clip,scale=30,keepaspectratio=true]{n9k3-all-pv0-05}&
%\includegraphics[width=1.9in, trim = 0.2mm 0mm 30mm 2mm, clip,scale=30,keepaspectratio=true]{n9k3-all-pv0-1}&
%\includegraphics[width=1.9in,trim = 0.2mm 0mm 30mm 2mm, clip,scale=30,keepaspectratio=true]{n9k3-all-pv0-2}\\
%\vspace{-.5em}\\
%(a)&(b)&(c)
%\vspace{1em}\\
%\end{array}$\vspace{-2em}
%\end{center}
%\caption{The percentage of positives with the KS test as a function of the history length of the KS test:
%(a) the upper bound threshold of the p.value of the KS test is $0.05$, (b) the upper bound threshold of the p.value of the KS test is $0.1$, (c) the upper bound threshold of the p.value of the KS test is $0.2$.
%From bottom to top the lines represent the behavior in a scenario with no cheaters, cheaters bidding values with a beta distribution where $\beta=0.9$, 
%cheaters bidding values with a beta distribution where $\beta=0.7$, and cheaters bidding values with a normal distribution ($\mu=0.5, \sigma=0.15$). }
%\label{ks-functioning}
%\end{figure*}

\deleted[]{The results of Fig~\ref{comparison} are obtained assuming a perfect GoF test, for comparison reasons with the other mechanisms and following the analysis.
In order to examine what is the overhead that a non-perfect (real) GoF test can impose we have chosen as 
an appropriate GoF test the Kologorov-Smirnov (KS) test. Fig~\ref{ks-functioning} demonstrates that the KS test is a good test for our mechanism. 
In this figure we evaluate the impact of the history length in the accuracy of the test. The history is the number of values the test keeps as a reference to decide 
if the presented new value belongs to a uniform distribution or not.}
\added[]{
First we examine the performance of the KS test when it is used as the GoF test in \fedor. Fig~\ref{ks-functioning} evaluates the impact of the history length in the accuracy of the test. As it can be seen the KS test is a good fit for our mechanism. The history length refers to the number of values the test keeps as a reference to decide whether the new value reported by the player follows the uniform distribution. 
}
Notice that, in all three graphs of Fig~\ref{ks-functioning}, in the case  where the player has a uniform distribution, as the history length gets smaller, the number of false positives increases. In addition to that, we have checked what happens in the case when the player cheats with a distribution that may resemble (or not) the uniform distribution. We have chosen the beta distributions with different parameter and a normal distribution. As you can see from the plots %of Figure~\ref{distribution} 
in Fig~\ref{distr07} and Fig~\ref{distr09}, with a beta distribution and a large $\beta$ parameter, the player can approximate the uniform behavior, and thus it is difficult for the KS test to identify a cheater. Notice that in all three graphs of Fig~\ref{ks-functioning}, as the length of the history test decreases, the percentage of true positives decreases.
For different values on the upper bound threshold of the p.value (which decides the acceptance of a value based on the history), we have noticed that using a $p.value<0.1$ is a good balance between having large percentage of true positives and not having too many false positives. For example, for a history length of $1000$ rounds (which is a value that does not give a great overhead to the mechanism), in Fig~\ref{ks-functioning}(a), where $p.value<0.05$, the true positives are close to $75\%$; in Fig~\ref{ks-functioning}(b), where $p.value<0.1$, the true positives are close to $82\%$; while in Fig~\ref{ks-functioning}(c), where the $p.value<0.2$, the true positives are close to $95\%$ while the false positive are much higher than in the other scenarios (close to $20\%$). 
Hence, selecting $1000$ rounds as the history of the mechanism, provides a good performance and is large enough if we compare it with the $10000$ rounds that our scenarios execute as part of the online auction process. In the rest of the section we will run experiments with a history length of $1000$ and KS test with $p.value<0.1$.
Before actual auctions are run, the history buffer is filled with $1000$ values, so the results are not affected by history's transient states.

We investigate now the utility of the honest players and cheaters in a variety of different scenarios, while having the KS test as our GoF test. 
In this study we assume that the players pay no flat fee.
We have shown analytically (assuming a perfect GoF test) that the best strategy for each player is to be honest. Our experimental results come to assert this even in the case where the GoF is not perfect. From Fig~\ref{utility}(a), the utility of the honest player is greater than the cheater in all scenarios considered. Independently of the behavior of the rest of the players, the mean utility of the honest player is around the same value of $5400$. There is no cheating strategy that will give a higher utility. Especially in scenarios $C$ and $G$ where the player cheats with a beta (where parameter is $\beta=0.9$) indeed the cheaters increase their mean utility by roughly around $500$ units but still the mean utility of the honest player is around $1500$ units higher. This also proves the efficiency of the KS test for $p.value<0.1$ and history length $1000$. The horizontal (blue) line in Fig~\ref{utility}(a) represents the value of utility as it is calculated from the analytical part. Notice that the analytical value is quite close to the experimental with the single exception of the player cheating with beta, where $\beta=0.9$.

Going one step further now we compare the utility of players when \fedor is used with the cases where VCG and GSP are used. As Fig~\ref{utility} shows, the honest player has always a larger utility compared with the utility of the cheating player in all scenarios considered, and in all mechanisms. 
However, in some scenarios the distance between the utility of a cheater and an honest player is very small (see, e.g., scenario $H$ in Fig~\ref{utility}(b)). 
It is also interesting to notice that in \fedor in scenario $A$, where all players are honest, the utility of the player is maximized compared to the honest players in the rest of the scenarios. This is not the case for the other two mechanisms (see scenario I for both VCG and GSP).

The utility of an honest player with VCG and GSP can vary significantly depending on the scenario. In any case, it is always around $1100$ for VCG and $700$ for GSP.
Comparing with \fedor, the utility of an honest player with the latter is up to four times larger than those with VCG and GSP. This is due to the fact that in our experiments we assumed that the flat fee of the player is zero. As we showed though in Fig~\ref{comparison}, the flat fee is something that can be adjusted to match the desired utility of the player (or the seller).

Finally, from Fig~\ref{social-utility} (a) we see that the social utility is maximized in \fedor when all players are truthful (scenario $A$ compared with the rest of scenarios). This does not hold for the other two mechanisms. As you can see from  Fig~\ref{social-utility} (b) and (c)  scenarios $B$ and $I$ (where dishonest players cheat with normal distribution $\mu=0.5, \sigma=0.15$) have greater social utility than scenario $A$ where they are all honest. 
%As we mentioned before this is because  \fedor mechanism ``reinforces'' the hehavior of the player inside the mechanism when a cheater is caught, while the other two mechanisms assuming rationality of the players, assume that players adapt there behavior.???? 
We can also notice that the social utility in \fedor is up to four times larger than those of GSP and VCG (but, of course, we are assuming no flat fee). 

\section*{Discussion}
\label{discussion}

\added[]{ Although powerful mechanisms (such as VCG) exist that, under monetary incentives, achieve the designed goals, cashing out payments is not always feasible. Maybe the system is distributed and no central authority exists that can guarantee the payments, or payments are too impractical. Another possibility is that the nature of the setting is such that payments are not allowed (i.e., a public common good). According to~\cite{arrow2012social, satterthwaite1975strategy} mechanisms without monetary incentives have limited capabilities. A case by case study is the only hope for finding mechanisms with good properties and without payments, like in~\cite{schmidinfocom, interdomain}. \fedor is a general mechanism that is truthful, fair and socially efficient, and without using monetary incentives. The inherent repeated nature of \fedor , though, makes it applicable only to settings with infinite or unknown number of repetitions. If the number of repetitions is known to the agents, then they can devise a strategy that would maximize their utility and potentially lead to a non-truthful mechanism. Our intuition is that, if the nature of the problem allows it,  mechanisms like \fedor can be devised for other problems as well. Taking advantage of the repeated interaction between the agents, and applying the concept of linking mechanisms, should allow to derive other mechanisms without payments for multiple combinations of desirable properties (i.e., different fairness criteria, strategy proofness, social efficiency, etc).} 
\added[]{An additional advantage, as we have seen experimentally, is that a mechanism without payments like \fedor was able, without incentivizing the workers through payments, to provide a way for the mechanism designer to select its utility and the utility of the participating players.  }

While \fedor is a mechanism with several interesting features, there are some issues that are worth discussing.
\fedor assumes that the valuations of the  players follow a uniform distribution. Although this might seem as a constraint, 
if the real bids do not follow such a distribution, they can be transformed by using the PIT transformation, as proposed by Santos et al.~\cite{QPQ}. \deleted[]{More over} \added[]{Moreover} we make the assumption that all players have i.i.d. valuations. As this is a first approach towards this line of mechanisms without monetary incentives we wanted to keep our model simple. In an extension of the mechanism, we plan to explore the correlated case where players distributions are correlated with those of other players~\cite{DBLP:conf/netys/SantosACF14}. 

An additional extension to the mechanism would be to allow players to provide their valuation for a particular good of the set and not for the whole set, as it is implemented now. This is yet another challenge, since players might have correlated preferences between certain goods in the same or different round.  

For the sponsored search scenario we have considered that all players pay the same flat fee. There could be the case that players with a higher budget would like to claim a larger number of slot over the consecutive rounds, something that \fedor does not allow to happen due to the fairness property. An approach to solve this constrain of \fedor would be to allow a player with a higher budget to participate with more than one identity. As a consequence, she will be increasing the amount of goods that she will receive, proportionally to the number of identities.  Again though, allowing players to have multiple identities one must assume that that players might have correlated preferences.

\noindent \paragraph{\bf Acknowledgments:}  
This work has been partially funded from research funds by the Regional Government of Madrid (CM) grant Cloud4BigData (S2013/ICE-2894) cofunded by FSE \& FEDER, the Spanish Ministry of Economy and Competitiveness grant HyperAdapt (TEC2014- 55713-
R), the Spanish Ministry of Education grant FPU2013-03792, the NSF of China
grant 61520106005 and partially funded by Real Colegio Complutense (RCC).

\newpage

\section*{Figures}

%% Include only the SI item label in the subsection heading. Use the \ref{label} command to cite SI items in the text.
%\subsection*{Figures:}
%\subsection*{S1}
%\label{comparison}
%{\bf Comparing \fedor with VCG and GSP in the mean utility of seller and player per auction.} 
%The values for \fedor form a line in each experiment, since the utilities of seller and players can be tuned with the flat fee.
%The values for VCG and GSP always lie on the corresponding \fedor line.
%%Flat fee of the player is zero and the waits of the slots/goods $k$ decrement from $k$ to $1$. 
%Plots represent the mean utility per auction of the seller and one player over 10000 rounds of execution and 100 experiments. (a) Scenarios with 9 player and the number of slots increasing from 1 (leftmost line) to 8 (rightmost line). (b) Scenarios with 3 slots and the players decreasing from 9 (line with largest slope) to 4 (line with smallest slope).

\begin{figure}[h!]
\includegraphics[width=40em]{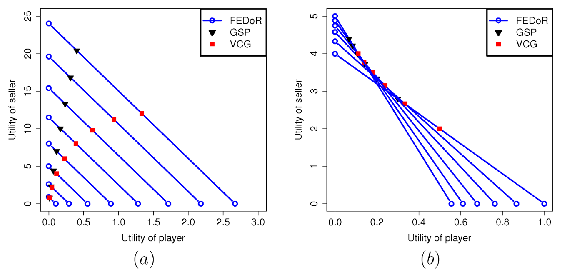}
\caption{{\bf Comparing \fedor with VCG and GSP in the mean utility of seller and player per auction.} 
The values for \fedor form a line in each experiment, since the utilities of seller and players can be tuned with the flat fee.
The values for VCG and GSP always lie on the corresponding \fedor line.
%Flat fee of the player is zero and the waits of the slots/goods $k$ decrement from $k$ to $1$. 
Plots represent the mean utility per auction of the seller and one player over $10000$ rounds of execution and $100$ experiments. (a) Scenarios with $9$ player and the number of slots increasing from $1$ (leftmost line) to $8$ (rightmost line). (b) Scenarios with $3$ slots and the players decreasing from $9$ (line with largest slope) to $4$ (line with smallest slope).}
\label{comparison}
\end{figure}

\begin{figure}[h!]
\includegraphics[width=40em]{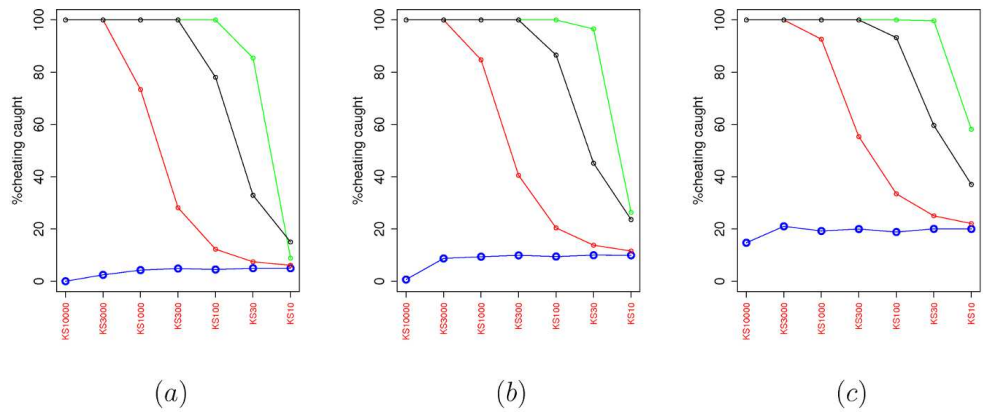}
\caption{{\bf The percentage of positives with the KS test as a function of the history length of the KS test.}
(a) The upper bound threshold of the p.value of the KS test is $0.05$, (b) the upper bound threshold of the p.value of the KS test is $0.1$, (c) the upper bound threshold of the p.value of the KS test is $0.2$.
From bottom to top the lines represent the behavior in a scenario with no cheaters, cheaters bidding values with a beta distribution where $\beta=0.9$, 
cheaters bidding values with a beta distribution where $\beta=0.7$, and cheaters bidding values with a normal distribution ($\mu=0.5, \sigma=0.15$).}
\label{ks-functioning}
\end{figure}

\begin{figure}[h!]
\includegraphics[width=40em]{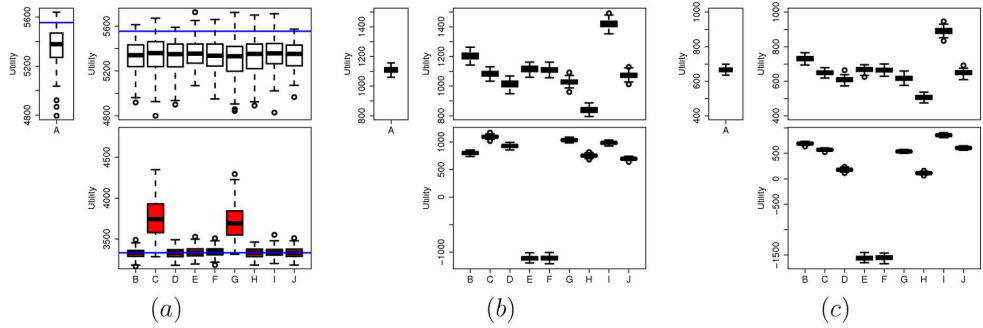}
\caption{{\bf The utility of an honest player compared to a cheating player in different scenarios.}
(a) Using the \fedor mechanism, for p.value$<0.1$ and history length $1000$ (the blue line represents the analytical utility values). (b) Using the VCG mechanism. (c) Using the GSP mechanism.
The left box for each mechanism (marked A) represents the distribution of the utility of an honest player in a scenario where all $9$ players are honest following a uniform. Right top box represents the utility of honest players that follow a uniform distribution in scenarios $B-J$. Left bottom box represents the utilities of cheating players with the distributions defined in scenarios $B-J$. Scenarios are as follows, $B$: $8$ honest players and one cheater with normal distribution ($\mu=0.5, \sigma=0.15$), $C$: $8$ honest players and one cheater with beta distribution with $\beta=0.9$, $D$: $8$ honest players and one cheater with beta distribution with $\beta=0.7$, $E$: $8$ honest players and one cheater with a random uniform distribution (different from its own distribution of values),  $F$: $6$ honest players and $3$ cheaters with random uniform distributions, 
$G$: $6$ honest players and $3$ cheaters with beta distributions ($\beta=0.9$), $H$: $6$ honest players and $3$ cheaters with beta distributions ($\beta=0.7$), $I$: $6$ honest players and $3$ cheaters with normal distributions ($\mu=0.5, \sigma=0.15$), $J$: $5$ honest players, $1$ cheater with random uniform distribution, $1$ cheater with beta distribution ($\beta=0.9$), $1$ cheater with beta distribution ($\beta=0.7$), and $1$ cheater with normal distribution ($\mu=0.5, \sigma=0.15$), which is the one represented in the plot.}
\label{utility}
\end{figure}

\begin{figure}[h!]
\includegraphics[width=40em]{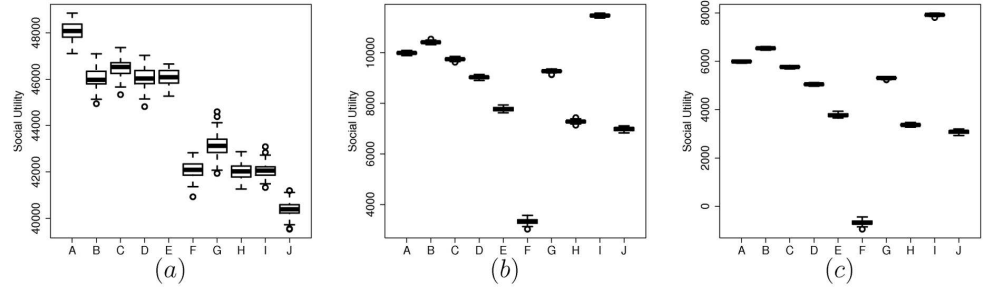}
\caption{{\bf The social utility in different scenarios.}
(a) Using the \fedor mechanism, for p.value<0.1 and history lenght $1000$ (b) Using the VCG mechanism. (c) Using the GSP mechanism.
The scenarios $A$ to $J$ are the same as in Fig~\ref{utility}.}
\label{social-utility}
\end{figure}

\begin{figure}[h!]
\includegraphics[width=40em]{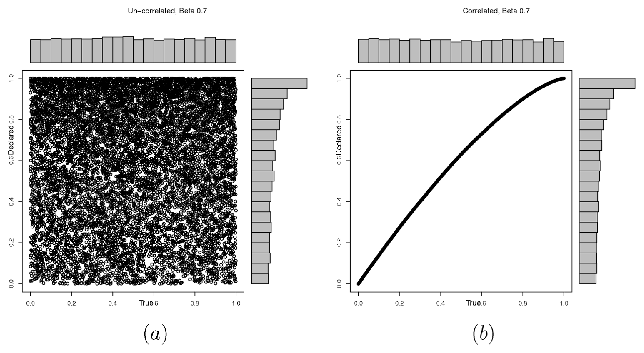}
\caption{{\bf Different strategies, modelled as a $\beta=0.7$ probability distribution.}(a) The declared values are not correlated with the true values (following a uniform distribution). (b)  The declared values are correlated with the true values (following a uniform distribution). Note that the marginal distribution represented in the x-axis corresponds to the declared values of the agent, while the y-axis represents the true values of the agent.} 
\label{distr07}
\end{figure}

\begin{figure}[h!]
\includegraphics[width=40em]{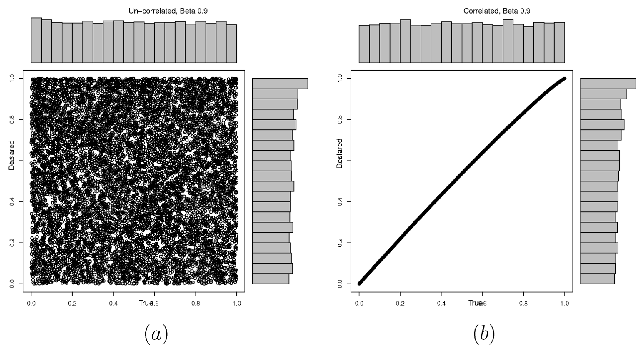}
\caption{{\bf Different strategies, modelled as a $\beta=0.9$ probability distribution.}(a) The declared values are not correlated with the true values (following a uniform distribution). (b)  The declared values are correlated with the true values (following a uniform distribution). Note that the marginal distribution represented in the x-axis corresponds to the declared values of the agent, while the y-axis represents the true values of the agent. }
\label{distr09}
\end{figure}

\clearpage
\newpage

\section*{Algorithms}

\lstset{columns=fullflexible,tabsize=3,basicstyle=\small,identifierstyle=\rmfamily\textit,
mathescape=true,morekeywords={process,const,when,elsif,procedure,null,function,do,return,
foreach,begin,var,if,then,else,rcv,rx_user_events,rx_network_events,send,upstream,downstream,
while,do,forward,backward,upon,wait,accept,audit,get},literate={:=}{{$\leftarrow$ }}1{->}{{ $\rightarrow$ }}1}
\lstset{escapeinside={('}{')}}

\begin{table}[H]
\centering

\begin{algorithm}[H]
\caption{\fedor Algorithm (code for player $i$)}
\begin{lstlisting}
foreach $k$ set of goods do
	Observe type $\theta_i$
	Choose type $\hat{\theta}_i$ using probability distribution 				$\sigma_i(\cdot|\theta_i)$ (player $i$'s strategy)
	Broadcast the type $\hat{\theta}_i$ (this is the bid of player $i$)
	Wait to receive the reported types $\hat{\theta}$ from all the players in $N$ 
	$d \leftarrow g_f(\hat{\theta})$              		\\ applying the FEDoR mechanism 
	if $(i=d_j)$ then player $i$ receives the $j$th slot
\end{lstlisting}
\label{alg3}
\end{algorithm} 
\end{table}

%\nolinenumbers
%\section*{References}
% Either type in your references using
% \begin{thebibliography}{}
% Text
% \end{thebibliography}
%
% OR
%
% Compile your BiBTeX database using our plos2015.bst
% style file and paste the contents of your .bbl file
% here.
% 

\end{document}